\documentclass[12pt]{article}
\usepackage[thmnum]{preamble}
\usepackage{subcaption}
\algnewcommand\algorithmicinput{\textbf{Input:}}
\algnewcommand\algorithmicoutput{\textbf{Output:}}
\algnewcommand\Input{\item[\algorithmicinput]}
\algnewcommand\Output{\item[\algorithmicoutput]}

\usepackage{graphicx}
\usepackage{tikz}
\usetikzlibrary{calc,shapes.multipart,chains,arrows,positioning,snakes}
\usepackage{url}
\urldef{\mailsa}\path|lalla,robere@cs.toronto.edu|  

\usepackage{color}

\begin{document}

\title{Path Graphs, Clique Trees, and Flowers \thanks{With apologies to Jack Edmonds \cite{11}.}}
\author{Lalla Mouatadid \quad \quad Robert Robere \\ Department of Computer Science \\ University of Toronto}

\maketitle
\begin{abstract}
  An \emph{asteroidal triple} is a set of three independent vertices in a graph such that any two vertices in the set are connected by a path which avoids the neighbourhood of the third.
  A classical result by Lekkerkerker and Boland \cite{6} showed that interval graphs are precisely the chordal graphs that do not have asteroidal triples. 
  Interval graphs are chordal, as are the \emph{directed path graphs} and the \emph{path graphs}.
  Similar to Lekkerkerker and Boland, Cameron, Ho\'{a}ng, and L\'{e}v\^{e}que \cite{4} gave a characterization of directed path graphs by a ``special type'' of asteroidal triple, and asked whether or not there was such a characterization for path graphs.
  We give strong evidence that asteroidal triples alone are insufficient to characterize the family of path graphs, and give a new characterization of path graphs via a forbidden induced subgraph family that we call \emph{sun systems}.
  Key to our new characterization is the study of \emph{asteroidal sets} in sun systems, which are a natural generalization of asteroidal triples.
  Our characterization of path graphs by forbidding sun systems also generalizes a characterization of directed path graphs by forbidding odd suns that was given by Chaplick et al.~\cite{9}.
\end{abstract}
\section{Introduction}
\label{sec:intro}

A graph $G$ is \emph{chordal} if every induced cycle in $G$ has at most three vertices.
Gavril \cite{5} proved that a graph $G$ is chordal if and only if $G$ can be represented as the intersection graph of a collection of subtrees of some tree $T$. 
This result suggests the definitions of some subfamilies of chordal graphs: a \emph{path graph} is the intersection graph of a collection of \emph{subpaths} on a tree $T$, 
a \emph{directed path graph} is the intersection graph of a collection of \emph{directed subpaths} on a \emph{directed tree} $T$, and an \emph{interval graph} is the intersection graph of a collection of subpaths on a path $P$.
It follows from the definitions and Gavril's result that we have the following sequence of containments: \[ \mathsf{Interval} \subset \mathsf{DirectPath} \subset \mathsf{Path} \subset \mathsf{Chordal}.\]
These containments are each strict, and minimal graphs exhibiting this are shown in Figure \ref{fig:3-4steps}.

\begin{figure}[htbp]
  \centering
  \resizebox{\linewidth}{!}{
    \begin{tikzpicture}
      \node[circle, draw, fill=black!100, inner sep=1pt, minimum width=4pt] (a) at (-1,0) {};
      \node[circle, draw, fill=black!100, inner sep=1pt, minimum width=4pt] (x) at (-2.5,0) {};
      \node[circle, draw, fill=black!100, inner sep=1pt, minimum width=4pt] (b) at (1,0) {};
      \node[circle, draw, fill=black!100, inner sep=1pt, minimum width=4pt] (z) at (2.5,0) {};
      \node[circle, draw, fill=black!100, inner sep=1pt, minimum width=4pt] (d) at (-1,1.5) {};
      \node[circle, draw, fill=black!100, inner sep=1pt, minimum width=4pt] (c) at (1, 1.5) {};
      \node[circle, draw, fill=black!100, inner sep=1pt, minimum width=4pt] (y) at (0, 3.0) {};
      \node[draw=none] at (0, -0.5) {$G_1$};
      \foreach \from/\to in {x/a, x/d, a/b, a/c, a/d, b/c, b/d, c/d, b/z, c/z, d/y, c/y}
      {\draw (\from) -- (\to);}

      \node[circle, draw, fill=black!100, inner sep=1pt, minimum width=4pt] (c) at (0+6, 0) {};
      \node[circle, draw, fill=black!100, inner sep=1pt, minimum width=4pt] (a) at (-1+6, 1.5) {};
      \node[circle, draw, fill=black!100, inner sep=1pt, minimum width=4pt] (b) at (1+6, 1.5) {};
      \node[circle, draw, fill=black!100, inner sep=1pt, minimum width=4pt] (x) at (0+6, 3) {};
      \node[circle, draw, fill=black!100, inner sep=1pt, minimum width=4pt] (y) at (2+6, 0) {};
      \node[circle, draw, fill=black!100, inner sep=1pt, minimum width=4pt] (z) at (-2+6, 0) {};
      \node[draw=none] at (0+6, -0.5) {$G_2$};
      \foreach \from/\to in {a/b, a/c, b/c, a/z, c/z, b/y, c/y, x/a, x/b}
      {\draw (\from) -- (\to);}

      \node[circle, draw, fill=black!100, inner sep=1pt, minimum width=4pt] (d) at (0+12, 0) {};
      \node[circle, draw, fill=black!100, inner sep=1pt, minimum width=4pt] (e) at (-1+12, 0) {};
      \node[circle, draw, fill=black!100, inner sep=1pt, minimum width=4pt] (z) at (-2.5+12, 0) {};
      \node[circle, draw, fill=black!100, inner sep=1pt, minimum width=4pt] (c) at (1+12, 0) {};
      \node[circle, draw, fill=black!100, inner sep=1pt, minimum width=4pt] (y) at (2.5+12, 0) {};
      \node[circle, draw, fill=black!100, inner sep=1pt, minimum width=4pt] (a) at (-1+12, 1.5) {};
      \node[circle, draw, fill=black!100, inner sep=1pt, minimum width=4pt] (b) at (1+12, 1.5) {};
      \node[circle, draw, fill=black!100, inner sep=1pt, minimum width=4pt] (x) at (0+12, 3.0) {};
      \node[draw=none] at (0+12, -0.5) {$G_3$};
      \foreach \from/\to in {a/b, a/e, a/d, a/c, b/e, b/d, b/c, e/d, d/c, a/z, e/z, b/y, c/y, a/x, b/x}
      {\draw (\from) -- (\to);}
    \end{tikzpicture}
  }
  \caption{$G_1$ is in $\mathsf{DirectPath} \setminus \mathsf{Interval}$, $G_2$ is in $\mathsf{Path} \setminus \mathsf{DirectPath}$, and $G_3$ is in $\mathsf{Chordal} \setminus \mathsf{Path}$}
  \label{fig:3-4steps}
\end{figure}
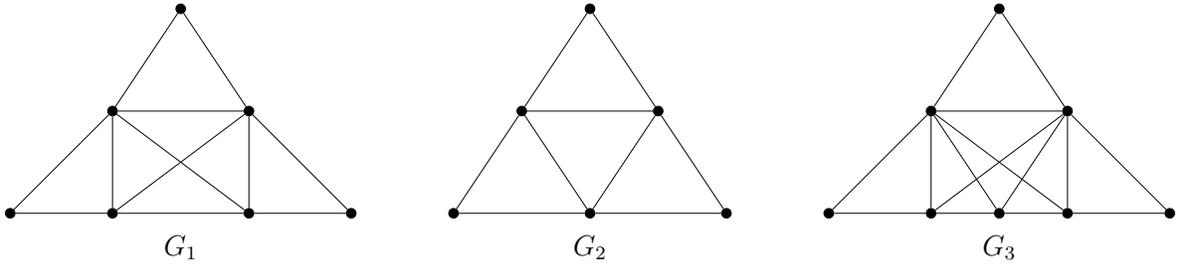

Another way of interpreting the definition of chordal graphs is as a list of \emph{minimal forbidden induced subgraphs}: a graph $G$ is chordal if it does not contain a $C_k$ as an induced subgraph for any $k \geq 4$.
The minimal forbidden induced subgraphs are also known for each of the families defined above: Lekkerker and Boland \cite{6} classified interval graphs in this way, directed path graphs were classified by Panda \cite{2}, and path graphs were classified L\'ev\^eque, Maffray and Preissman \cite{3} (see Figure \ref{fig:FISC} for the minimal forbidden induced subgraphs of the path graphs).

If $G$ is a graph, a set of three distinct and independent vertices $x, y, z$ in $G$ is called an \emph{asteroidal triple} if any pair of vertices $\alpha, \beta \in \set{x, y, z}$ remain connected in $G$ if we delete the third vertex and its neighbourhood.
These triples play an important role in the study of interval graphs: Lekkerkerker and Boland \cite{6} derived the list of minimal forbidden induced subgraphs of interval graphs by first proving the following (now classic) theorem.
\begin{thm}
A chordal graph $G$ is an interval graph if and only if it does not contain an asteroidal triple.
\end{thm}
Interestingly, both Panda and L\'ev\^eque et al.~\cite{2, 3} list the forbidden subgraphs of directed path and path graphs respectively by direct proofs, leaving the characterization of these classes in terms of asteroidal triples open. 
For directed path graphs, this problem was resolved by Cameron et al.~\cite{4}, who gave a comparable theorem to that of Lekkerkerker and Boland:
\begin{thm}
  A chordal graph $G$ is a directed path graph if and only if it does not contain a special asteroidal triple.
\end{thm}
A special asteroidal triple is an asteroidal triple where each pair of vertices in the triple must be connected by some special subgraph: see Section \ref{sec:special} for a further discussion of this result.
Notably, a characterization for path graphs via asteroidal triples of this type was explicitly left open by both \cite{4, 3}.

We fill this gap.
For $k \geq 3$, recall that a graph $G(V,E)$ is a $k$-\emph{sun} if $V = C\cup R$ where $C = \{c_i \st i \in [k]\}$ is a clique (the \emph{core} of the sun), $R = \{r_i \st i \in [k]\}$ is an independent set (the \emph{rays} of the sun), and for all $r_i \in R$, $N(r_i) = \{c_i, c_{i+1}\}$ mod $k$. 
 The graph $G_2$ in Figure \ref{fig:3-4steps} is the 3-sun. 
 We define a generalization of the sun graph which we call a \emph{sun system} (Definition \ref{def:sun-system}), and show that a graph is a path graph if and only if it does not contain a ``bad'' sun system in a specific technical sense (Theorem \ref{thm:main}).
 A sun system can be viewed as an \emph{asteroidal set} of vertices of possibly unbounded size which are mutually connected by some specific mediating structure, and in this sense we generalize the result of Cameron et al.~\cite{4}.
 This new characterization also generalizes a result of Chaplick et al.~\cite{9} characterizing directed path graphs as path graphs which do not contain an odd sun as an induced subgraph.
 We then further discuss why characterizing path graphs solely by some ``special'' type of asteroidal triple will fail, as asteroidal triples alone are not enough to provide a certificate that a chordal graph is not a path graph on the underlying tree (Proposition \ref{prop:at_obstruction}). 

The outline of the rest of the paper is as follows.
In Section \ref{sec:defn} we give the necessary preliminaries and state our main theorem.
Section \ref{sec:main} proves our main theorem, and Section \ref{sec:special} discusses how asteroidal triples are insufficient to characterize path graphs.
In Section \ref{sec:conclusion} we discuss open problems suggested by this work.

\section{Definitions}
\label{sec:defn}

If $a, b$ are integers with $a < b$ we let $[a, b] := \set{a, a+1, \ldots, b-1, b}$.
If $A$ and $B$ are sets we write $A \bowtie B$ if both $A \not \subseteq B$ and $B \not \subseteq A$ holds.

\begin{defn}\label{def:clique-tree}
  Let $G$ be a graph.
  A tree $T$ is called a \emph{clique tree} of $G$ if the vertices of $T$ can be labelled with the maximal cliques of $G$ such that the following holds: if $v$ is any vertex in $G$ then the subgraph $T^v$ of $T$ induced by the set of cliques containing $v$ is connected.
  Furthermore, $T$ is a \emph{clique-path tree} if $T$ is a clique tree and for every vertex $v$ of $G$ the subgraph $T^v$ is a path.
\end{defn}

As mentioned in the introduction, Gavril \cite{5} proved that a graph $G$ is chordal if and only if $G$ has a clique tree, where the subtrees of the clique tree give the intersection model of $G$. Similarly, Gavril also showed that a graph $G$ is a path graph if and only if $G$ has a clique-path tree \cite{7}.
If $T$ is a clique tree of a graph $G$ and $S$ is a set of vertices then we use $T[S]$ to denote the subtree of $T$ of minimum size such that its set of vertices contains $S$.
That is, for every clique node $Q$ in $T[S]$, $Q \cap S \neq \emptyset$. 
If $S_1, S_2, \ldots, S_t$ are a collection of subsets of vertices then let \[ T[S_1, S_2, \ldots, S_t] := T[S_1 \cup S_2 \cup \cdots \cup S_t].\]

L{\'{e}}v{\^{e}}que et al.~gave the family of forbidden induced subgraphs for path graphs, which we reproduce in Figure \ref{fig:FISC}.

\begin{figure}[htbp]
  \centering
  \includegraphics[scale=0.4]{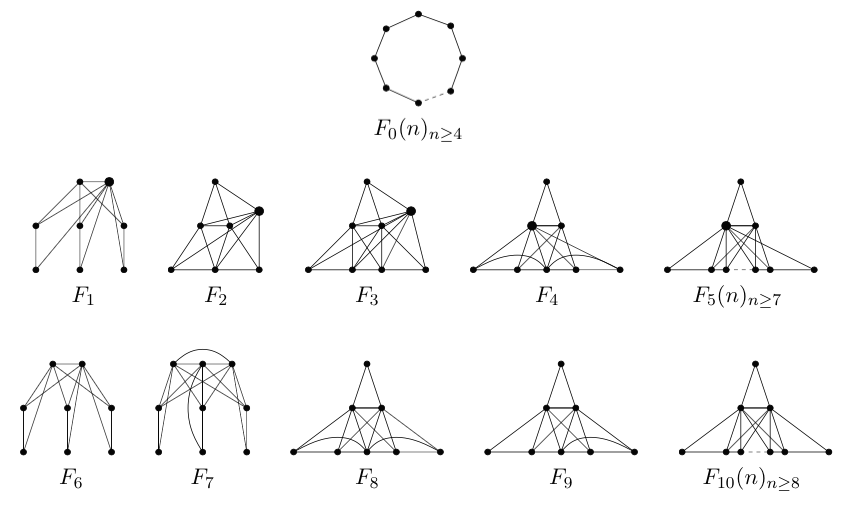}
  \includegraphics[scale=0.4]{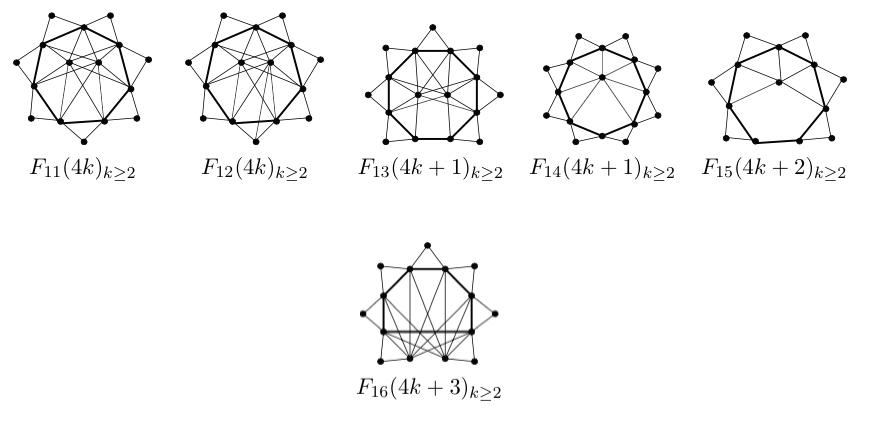}
  \caption{The forbidden induced subgraphs of path graphs \cite{3}. A cycle with bold edges indicates a clique.}
  \label{fig:FISC}
\end{figure}
\begin{defn}\label{def:asteroidal-set}
A set of vertices $S \subseteq V$ is \emph{asteroidal} in $G$ if for every $v \in S$, the set $S \setminus \set{v}$ is in the same connected component after removing $N[v]$ from $G$.
We only consider asteroidal sets with at least three vertices in order to avoid trivial cases, and so if $|S| = 3$ (the minimal case) then $S$ is an \emph{asteroidal triple}.
\end{defn}

The following result is folklore (see \cite{1} for a proof).
\begin{lem}\label{lem:asteroidal-leaves}
  Let $G$ be a chordal graph and let $u, v, w$ be an asteroidal triple in $G$.
  Then for each vertex $\alpha \in \set{u, v, w}$ there is a unique maximal clique $Q_\alpha$ containing $\alpha$ such that in any clique tree $T$ of $G$, $T[u,v,w]$ has exactly three leaves $Q_u, Q_v, Q_w$.
\end{lem}

An immediate consequence of this lemma is the next proposition, which we will need later.
\begin{prop}\label{prop:path-graph-adjacency}
  Let $G$ be a chordal graph and let $u, v, w$ be an asteroidal triple in $G$.
  If there is a vertex $x$ such that $u, v, w \in N(x)$ then $G$ is not a path graph.
\end{prop}

 Recall that a collection $\mathcal{S}$ of sets $S_i$ satisfies the \emph{Helly property} if for every subset $T \subseteq S$ and every pair $S_i,S_j \in T$, $S_i\cap S_j \neq \emptyset$ implies $\bigcap\limits_{S_i\in T} S_i \neq \emptyset$.
 Our next goal is to define \emph{sun systems}, which is used in our characterization of path graphs. 
 To do so, we first introduce a \emph{flower}.
 \begin{defn}
 A \emph{flower}
 $\mathcal{F} = \{P_1, \ldots, P_t\}$ is a collection of cliques that satisfies the Helly property such that $\bigcap\limits_{P_i \in \mathcal{F}} P_i \neq \emptyset$. We refer to $\bigcap\limits_{P_i \in \mathcal{F}} P_i = C$ as the \emph{core} of the flower, and the clique $P_i \setminus C$ as a \emph{petal} of the flower $\mathcal{F}$\footnote{Note the connection between a flower and a sunflower (made famous by Erd\H{o}s and Rado \cite{10}). 
To contrast, a sunflower is a flower where all the petals are not necessarily cliques but must satisfy the stronger condition \[(P \setminus C) \cap (P' \setminus C) = \emptyset \] instead of $P \bowtie P'$.}. 
 \end{defn}
 \begin{defn}\label{def:sun-system}
   A graph $G$ is a \emph{sun system} if the vertices of $G$ can be partitioned as $V = F \cup R$ such that $R$ is asteroidal and the induced subgraph of $G$ on $F$ is a flower ${\cal F}$; $G$ is \emph{non-trivial} if $|{\cal F}| > 1$.
   The asteroidal vertices in $R$ are called \emph{rays}.


 \end{defn}

A sun system is just a $k$-sun where the central clique has been replaced with a flower of cliques and where the ``rays'' no longer need to be connected cyclically to the core (cf. Figure \ref{fig:Hsun}).
Suppose that $G$ is a sun system on a flower $\mathcal{F}$. Then since the set of rays $R$ is asteroidal, it follows that for all $r \in R$ we have $N(r) \subseteq \mathcal{F}$, and moreover $N(r) \bowtie N(r')$ for all $r, r' \in R$. 
In particular, no ray is in the neighbourhood of any other ray.

\begin{figure}[htbp]
  \centering
  \begin{minipage}{0.45\linewidth}
  \centering
  \includegraphics[scale=0.2]{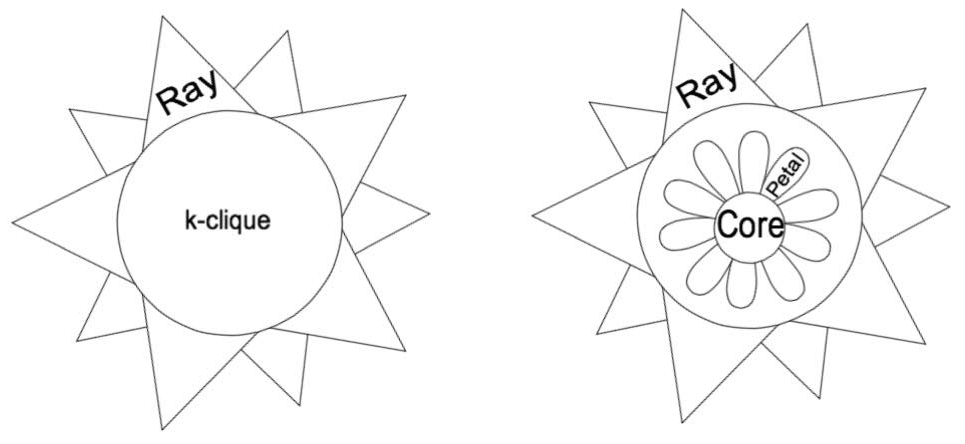}
  \end{minipage}
  \caption{A sun and a sun system.}
  \label{fig:Hsun}
\end{figure}


Figure \ref{fig:Hsun} illustrates a ``generic'' sun system, while the graph in Figure \ref{fig:clique-example} is a concrete example of a non-trivial sun system. In particular, the flower of Figure \ref{fig:clique-example} is formed by two maximal cliques (indicated by vertices $x$ and $y$ and the core outlined by the clique $(u,v,w)$), and three rays which are each adjacent to vertices in the core.

\begin{figure}[htbp]
  \centering
  \begin{tikzpicture}
    \node[circle, draw, fill=black!100, inner sep=1pt, minimum width=4pt] (x) at (-1.5, 1.5) {};
    \node[circle, draw, fill=black!100, inner sep=1pt, minimum width=4pt] (a) at (-1,-0.5) {};
    \node[circle, draw, fill=black!100, inner sep=1pt, minimum width=4pt] (b) at (0,1) {};
    \node[circle, draw, fill=black!100, inner sep=1pt, minimum width=4pt] (z) at (0,-2.12) {};
    \node[circle, draw, fill=black!100, inner sep=1pt, minimum width=4pt] (c) at (1,-0.5) {};
    \node[circle, draw, fill=black!100, inner sep=1pt, minimum width=4pt] (y) at (1.5, 1.5) {};
    \node[circle, draw, fill=black!100, inner sep=1pt, minimum width=4pt] (g) at (-3, 0.25) {};
    \node[circle, draw, fill=black!100, inner sep=1pt, minimum width=4pt] (h) at (3, 0.25) {};
    \foreach \from/\to in {x/a, x/b, b/a, z/a, z/c, a/c, y/b, y/c, b/c, a/g, b/g, c/g, a/h, b/h, c/h}
    {\draw (\from) -- (\to);}

    \pgfmathsetmacro{\xoffset}{8}
    \node[circle, draw, fill=black!100, inner sep=1pt, minimum width=4pt] (x1) at (-1.5 + \xoffset, 1.5) {};
    \node[circle, draw, inner sep=1pt, minimum width=4pt] (a1) at (-1 + \xoffset,-0.5) {u};
    \node[circle, draw,  inner sep=1pt, minimum width=4pt] (b1) at (0 + \xoffset,1) {v};
    \node[circle, draw, fill=black!100, inner sep=1pt, minimum width=4pt] (z1) at (0 + \xoffset,-2.12) {};
    \node[circle, draw, inner sep=1pt, minimum width=4pt] (c1) at (1 + \xoffset,-0.5) {w};
    \node[circle, draw, fill=black!100, inner sep=1pt, minimum width=4pt] (y1) at (1.5 + \xoffset, 1.5) {};
    \node[circle, draw, inner sep=0pt, minimum width=5pt] (g1) at (-3 + \xoffset, 0.25) {x};
    \node[circle, draw,  inner sep=0pt, minimum width=5pt] (h1) at (3 + \xoffset, 0.25) {y};

    \foreach \from/\to in {a1/b1, a1/c1, b1/c1}
    {\draw[very thick, red] (\from) --(\to);}
    \foreach \from/\to in {x1/a1, x1/b1, z1/a1, z1/c1, y1/b1, y1/c1, a1/g1, b1/g1, c1/g1, a1/h1, b1/h1, c1/h1}
    {\draw (\from) -- (\to);}
  \end{tikzpicture}
  \caption{The graph $F_{11}(8)$, a sun system from Figure 2, with its flower outlined on the right.}
  \label{fig:clique-example}
\end{figure}
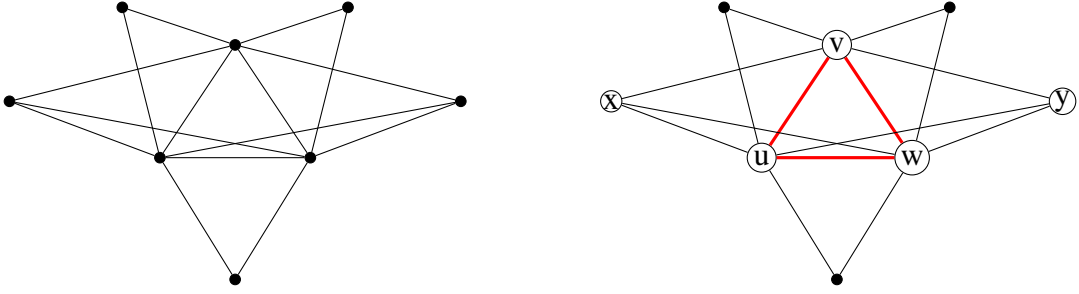

 In a sun system the rays can be adjacent to vertices in the core or vertices in the petals of the flower. 
 For our characterization we will only need to consider rays with neighbourhoods of a certain type.

\begin{defn}\label{def:ray-types}
  Let $G$ be a sun system on a flower $\mathcal{F} = \{P_1, \ldots, P_t\}$ with core $C$.
  A ray $r$ of $G$ is called \emph{intersecting} if $N(r) \subset C$, and \emph{split on $P_r$} if there is a \emph{unique} petal $P_r$ in the underlying flower $\mathcal{F}$ such that $N(r) \subseteq P_r$ and both $N(r) \cap C \neq \emptyset$, $N(r) \cap (P_r\setminus C) \neq \emptyset$.
  We say that $P \in \mathcal{F}$ is \emph{split} if at least one ray is split on $P$.
\end{defn}

For example, in the sun system outlined in Figure \ref{fig:clique-example}, all of the ray vertices are intersecting.
These types of rays are useful for two reasons; first, if $r$ is an intersecting or split ray then $N[r]$ is a maximal clique in $G$ (and so $N[r]$ will appear as a vertex in the clique tree of $G$); second, if $r$ is a split ray on a petal $P$ then it can be shown that the maximal cliques $N[r]$ and $P$ must be adjacent in any clique tree of the sun system.
Therefore, from here on in the paper we make the following assumption:
\begin{center}
  All ray vertices in any sun system are intersecting or split.
\end{center}
We may make this assumption since a graph $G$ will be a path graph if and only if it does not contain a ``bad'' induced sun system, and so other ray vertices can be safely ignored.

There are other useful consequences of the previous definition.
In particular, if two ray vertices $r, r'$ satisfy $v \in N(r) \cap N(r')$ for some $v$ then the corresponding maximal cliques $N[r], N[r']$ must be connected in the subtree $T^v$ of any clique tree $T$ for $G$.
In our characterization, sequences of ray vertices $r_1, r_2, \ldots, r_t$ that satisfy $N(r_i) \cap N(r_{i+1}) \not = \emptyset$ play an important role for this very reason.
It will be convenient to introduce a graph which captures these sequences of ``neighbourhood-adjacent'' ray vertices.

\begin{defn}\label{def:auxiliary-graph}
  Let $G$ be a sun system with flower $\mathcal{F}$ and rays $R$.
  Let $P_1, P_2, \ldots, P_t$ and $C$ denote the petals and the core respectively of the flower $\mathcal{F}$.
  We define a new graph $A^G$, called the \emph{auxiliary graph of $G$}, as follows.

  For each ray $r \in R$, we introduce a new \emph{ray vertex} $\overline r \in A^G$, and for each \emph{split} petal $P_i$, we introduce a \emph{petal vertex} $\overline p_i$.
  The edges of $A^G$ are as follows:
  \begin{itemize}
  \item Add an edge between each pair of petal vertices $\overline p, \overline p'$.
  \item For each pair of ray vertices $\overline r, \overline r'$ add an edge if $N(r) \cap N(r') \cap C \not = \emptyset$.
  \item For each pair $\overline r, \overline p$, add an edge if the ray $r$ is split on the corresponding petal $P$.
  \end{itemize}
\end{defn}

Now we state our main theorem.

\newtheorem*{thm:main}{{\bf Theorem \ref{thm:main}}}
\begin{thm}\label{thm:main}
  A chordal graph $G$ is a path graph if and only if it does not contain an induced, non-trivial sun system $G$ where $A^G$ is non-bipartite.
\end{thm}
 A ''bad" sun system is therefore a sun system $G$ whose auxiliary graph $A^G$ is non-bipartite.

\section{Main Result}
\label{sec:main}
 The proof of Theorem \ref{thm:main} is by contradiction.
So, assume that we have a chordal graph which contains a non-trivial, induced sun-system $G$ such that $A^G$ is not bipartite.
Let $O$ be an odd cycle in $A^G$, and we split into three cases depending on what types of vertices lie on $O$.
Before we begin in earnest, we prove two preliminary lemmas that will prove useful later.

\begin{lem}\label{lem:two-cliques}
  Let $G$ be a non-trivial sun system and let $R' \subseteq R$ be a set of rays in $G$ such that all rays in $R'$ are split or intersecting.
  Let $C$ denote the core of the underlying flower of $G$.
  Suppose that $G$ is a path graph.
  Then there exist two maximal cliques $Q_1, Q_2$ in $G$ such that $C \subseteq Q_1 \cap Q_2$, and in any clique-path tree of $G$ and for every $r \in R'$ the maximal clique $N[r]$ is adjacent to either $Q_1$ or $Q_2$.
\end{lem}
\begin{proof}
  Suppose not by way of contradiction, and let $T$ be any clique-path tree of $G$.
  If $Q$ is a maximal clique in $G$ that does \emph{not} contain a ray vertex $r$, then since $G$ is a sun-system there must be a collection of cliques $\set{P_i}_{i \in I} \subseteq \mathcal{F}$ for some index set $I$ such that \[Q = \bigcap_{i \in I} P_i.\]
  It follows that $C \subseteq Q$.
  
  Now, since each ray $r \in R'$ is split or intersecting, the set $N[r]$ is a maximal clique in $G$.
  So suppose that $r_1, r_2, r_3$ are three ray vertices such that $N[r_1], N[r_2], N[r_3]$ are adjacent to three distinct maximal cliques $Q_1, Q_2, Q_3$, respectively, in $T$.
  By Lemma \ref{lem:asteroidal-leaves} these three cliques $N[r_1], N[r_2], N[r_3]$ are leaves in $T[r_1, r_2, r_3]$, and so they can not be adjacent to each other (or to any other maximal clique in $G$ containing a ray vertex).
  Thus each of the cliques $Q_1, Q_2, Q_3$ must contain the core $C$ of the underlying flower $\mathcal{F}$ of $G$.
  
  Let $\mathcal{P}$ denote the path in $T$ connecting the three nodes $Q_1, Q_2, Q_3$ --- note that it must be a path, since $C \subseteq Q_1 \cap Q_2 \cap Q_3$ and $T$ is a clique-path tree --- and assume by symmetry that $Q_2$ is not one of the endpoints of $\mathcal{P}$.
  Towards contradiction, choose any vertex $v \in N[r_2] \cap C$: it follows that $v \in Q_1 \cap Q_2 \cap Q_3 \cap N[r_2]$, and so the tree $T^v$ is not a path.
\end{proof}

\begin{lem}\label{lem:split-rays}
  Let $G$ be a non-trivial sun system with underlying flower $\mathcal{F}$, let $P$ be a petal in $\mathcal{F}$, and let $r$ be a ray of $G$ that is split on $P$.
  Then in any clique tree $T$ of $G$ the clique $N[r]$ is adjacent to $P$.
\end{lem}
\begin{proof}
  Let $R$ be the set of asteroidal rays in $G$.
  By the definition of a split ray, we know $r$ has neighbours in exactly one petal $P$, and all other neighbours of $r$ lie in the core of $\mathcal{F}$. 
  Let $p \in P$ be a neighbour of $r$ in $P$.
  Let $T$ be any clique tree of $G$.
  Since $p$ is contained in $P$ and no other petal, it follows that any maximal clique that contains $p$ must be either the petal $P$ or a maximal clique containing some ray vertex adjacent to $p$.
  This means that if $N[r]$ is not adjacent to $P$, then it must be adjacent to some maximal clique $N[r']$ containing the vertex $p$.
  But since $r$ and $r'$ are both asteroidal, it follows that in the clique tree $T[R]$ the cliques $N[r]$ and $N[r']$ are both leaves, which is a contradiction.
\end{proof}

We first consider the case where the odd cycle $O$ consists solely of petal vertices.
First, note that the petal vertices of $A^G$ form a clique.
Thus if $A^G$ contains at least three petal vertices then $A^G$ contains a triangle as an induced cycle in a clique on $|O| \geq 3$ vertices.
\begin{lem}\label{lem:cycle-petal}
  Let $G$ be a non-trivial sun system such that $A^G$ has three or more petal vertices.
  Then $G$ is not a path graph.
\end{lem}
\begin{proof}
  Let $G$ be a non-trivial sun system such that $A^G$ has three or more petal vertices, and assume by way of contradiction that $G$ is a path graph.
  Let $\mathcal{F}$ denote the underlying flower of the sun system $G$, and let $P_1, P_2, P_3$ denote the three split petals in $\mathcal{F}$ guaranteed by assumption.
  It follows that there exists three ray vertices $r_1, r_2, r_3$ such that $r_i$ is split on $P_i$ for each $i = 1, 2, 3$, and note that the definition of a split vertex (Definition \ref{def:ray-types}) implies that the three ray vertices are distinct. 
  By the definition of a sun system we know that the three vertices $r_1, r_2, r_3$ are asteroidal.
  
  Let $T$ be any clique-path tree of $G$.
  By definition of a flower, the sets $P_1, P_2, P_3$ are each maximal cliques in $G$ and so it follows that they are vertices in the clique tree.
  By Lemma \ref{lem:split-rays} the sets $N_i$ and $P_i$ are adjacent in the clique tree for all $i$, so applying Lemma \ref{lem:two-cliques} yields a contradiction since $N_1, N_2, N_3$ must each be adjacent to one of two maximal cliques.
\end{proof}

\begin{cor}\label{cor:cycle-petal}
  Let $G$ be a non-trivial sun system such that $A^G$ has an odd cycle $O$ consisting entirely of petal vertices. 
  Then $G$ is not a path graph.
\end{cor}

Next we consider the case where the odd cycle $O$ consists only of ray vertices.
We show that either three of the rays must share a core vertex in their neighbourhood (an easy contradiction against Proposition \ref{prop:path-graph-adjacency}), or a certain type of ``parity obstruction'' exists in the clique tree.
\begin{lem}\label{lem:cycle-rays}
  Let $G$ be a non-trivial sun-system such that $A^G$ has an odd cycle $O$ consisting only of ray vertices.
  Then $G$ is not a path graph.
\end{lem}
\begin{proof}
  Assume that $G$ is a path graph by way of contradiction, and let $O = \set{r_0, r_1, \ldots, r_{t-1}}$, ordered so that $r_i, r_{i+1\pmod t}$ are each adjacent for all $i$ in $A^G$.
  By Proposition \ref{prop:path-graph-adjacency} we know that there is no vertex $v \in G$ such that the neighbourhood of $v$ contains three distinct rays.
  In particular, this implies that every core vertex is adjacent to at most two rays in $O$.
  Let $T$ be a clique-path tree of $G$.
  By Lemma \ref{lem:two-cliques} there are two maximal cliques $Q_0, Q_1$ in $T$ such that the maximal clique $N[r_i]$ is adjacent to either $Q_0$ or $Q_1$ in $T$ for all $i \in [0,t-1]$.
  
  If $r, r' \in O$ are adjacent rays then $N[r], N[r']$ cannot be adjacent to the same maximal clique in $T$.
  To see this, let $v \in N[r] \cap N[r'] \cap C$, and assume w.l.o.g.~that $N[r], N[r']$ are both connected to $Q_0$.
  Then $v \in Q_0 \cap Q_1$ and so it follows that $T^v$ is not a path, which is a contradiction.

  So, suppose w.l.o.g. that $N[r_0]$ is adjacent to $Q_0$.
  Then the previous fact combined with an easy induction shows that $N[r_i]$ is adjacent to $Q_j$ where $j \equiv i \pmod 2$.
  Since $t$ is odd, $t-1$ is even, and so $N[r_0]$ and $N[r_t]$ are both adjacent to $Q_0$, which is a contradiction.
\end{proof}

Finally we consider the case where the odd cycle can touch both petal and ray vertices.
Once again we show that the ``parity obstruction'' must exist.
\begin{lem}\label{lem:cycle-both}
  Let $G$ be a non-trivial sun-system such that $A^G$ has an odd-cycle $O$ containing petal and ray vertices.
  Then $G$ is not a path graph.
\end{lem}
\begin{proof}
  Assume not, by way of contradiction, and let $T$ be a clique-path tree of $G$.
  Let $\mathcal{F}$ denote the flower of $G$.
  By Lemma \ref{lem:cycle-petal} we may assume that $A^G$ has at most two petal vertices.
  We break into two sub-cases.

  \spcnoindent
  {\bf Case 1:} $O$ contains exactly one petal vertex.

  Let $p$ be the unique petal vertex in $O$, and let $r_0, r_1, \ldots, r_{t-1}$ be the ray vertices.
  Write $O = \set{p, r_0, r_1, \ldots, r_{t-1}}$ where $r_i$ and $r_{i+1 \pmod t}$ are adjacent for all $i$, and where $p$ is adjacent to $r_0$ and $r_{t-1}$.
  Let $P$ be the petal in $\mathcal{F}$ corresponding to the petal vertex $p$.
  Since both $N[r_0]$ and $N[r_{t-1}]$ are adjacent to $p$, it follows that the rays $r_0$ and $r_{t-1}$ are both split on $P$, which by Lemma \ref{lem:split-rays} implies that $N[r_0]$ and $N[r_{t-1}]$ are both connected to $P$ in $T$.
  By Lemma \ref{lem:two-cliques}, there exists two maximal cliques $Q_0, Q_1$ containing the core of the flower $\mathcal{F}$ such that $N[r_i]$ is adjacent to either $Q_0$ or $Q_1$ for all $i$.
  Clearly one of these two cliques must be $P$, so assume $P = Q_0$.
  
  As we argued in Lemma \ref{lem:cycle-rays}, we have that $N[r_i]$ and $N[r_{i+1}]$ are not adjacent to the same clique in $T$ for each $i$.
  Thus, if $N[r_0]$ is connected to $Q_0$, an easy induction implies that $N[r_i]$ is adjacent to $Q_j$ for $j \equiv i \pmod 2$.
  Since there are an even number of ray vertices, it follows that $t-1$ is odd, and so $N[t-1]$ is adjacent to $Q_1$, a contradiction.

  \spcnoindent
  {\bf Case 2:} $O$ contains exactly two petal vertices.
  
  This case is similar to the previous case.
  Now, let $p_0, p_1$ be the two petal vertices and again let $r_0, r_1, \ldots, r_{t-1}$ be the ray vertices.
  Write $O = \set{p_0, r_0, r_1, \ldots, r_{t-1}, p_1}$.
  Let $P_i$ be the petal in the flower $\mathcal{F}$ underlying $G$ corresponding to the petal vertex $p_i$.
  Since $r_0$ is adjacent to $p_0$ and $r_{t-1}$ is adjacent to $p_1$, Lemma \ref{lem:split-rays} implies that $N[r_0]$ is adjacent to $P_0$ and $N[r_{t-1}]$ is adjacent to $P_1$.
  By Lemma \ref{lem:two-cliques}, there exists two maximal cliques $Q_0, Q_1$ containing the core of the flower $\mathcal{F}$ such that $N[r_i]$ is adjacent to either $Q_0$ or $Q_1$ for all $i$.
  These two cliques must therefore be $P_0, P_1$, and so let $Q_0 = P_0, Q_1 = P_1$.
  Note again that $N[r_0]$ is adjacent to $Q_0$ and $N[r_{t-1}]$ is adjacent to $Q_1$.
  
  As we argued in Lemma \ref{lem:cycle-rays}, we have that $N[r_i]$ and $N[r_{i+1}]$ are not adjacent to the same clique in $T$ for each $i$.
  Thus, if $N[r_0]$ is connected to $Q_0$, an easy induction implies that $N[r_i]$ is adjacent to $Q_j$ for $j \equiv i \pmod 2$.
  Since there are an odd number of ray vertices in $O$, it follows that $t-1$ is even, and so $N[t-1]$ must be adjacent to $Q_0$, a contradiction.
\end{proof}

Theorem \ref{thm:main} follows as an easy consequence of these three lemmas.

\begin{proof}[Proof of Theorem \ref{thm:main}]
  By way of contradiction, suppose that $G$ is a path graph that contains a non-trivial, induced sun system $G$ where $A^G$ is non-bipartite.
  Then $A^G$ must contain an odd-cycle $O$ which lies completely on the petals, completely on the rays, or on both.
  However, each of these cases are impossible by Corollary \ref{cor:cycle-petal} and Lemmas \ref{lem:cycle-rays}, \ref{lem:cycle-both}, respectively.
  Thus $A^G$ must be bipartite.

  For the reverse direction, it is a simple matter to check each of the forbidden induced subgraphs for such a ``bad'' sun-system $G$.
  Rather than do this for each of the forbidden induced subgraphs, we collect the families of forbidden induced subgraphs from Figure \ref{fig:FISC} into sets according to where the odd cycle appears in the auxiliary graph $A^G$.
  \begin{enumerate}
  \item In $F_1, F_2, F_3, F_4, F_5$, the asteroidal triples are each adjacent to a universal vertex. 
    It is easy to construct a sun system with the asteroidal triples as rays and the universal vertex in the core; this yields an odd cycle on the ray vertices in $A^G$.
  \item In $F_6, F_7$, the top layer of vertices is the core, the second layer of vertices form three petals, and the third layer of vertices are the rays.
    The rays are therefore split on three distinct petals, which forms an odd cycle on the petal vertices in $A^G$
  \item In $F_{11}, F_{12}$, there are two petals which overlap in the large central clique.
    Each of the rays are attached to the central clique --- this forms an odd cycle on the ray vertices in $A^G$.
  \item In $F_{14}, F_{15}$, we have two central cliques (one containing the central vertex, and one that does not).
    Each of these cliques form a petal, and the two bottom rays are split on a single petal.
    This forms an odd cycle on the rays and petals (with a single split petal) on $A^G$.
  \item In $F_{13}$ and $F_{16}$ there are two central cliques corresponding to the central vertices that form the two petals.
    The two bottom rays are each split on different petals, and so this forms an odd cycle on the rays and petals (with two split petals) on $A^G$.
  \end{enumerate}
\end{proof}

\section{Remarks on Asteroids and Special Connections}
\label{sec:special}

We next discuss why asteroidal triples alone are not enough to characterize path graphs. 
Cameron et al.\cite{4} gave an interesting characterization of directed path graphs by forbidding asteroidal triples with some ``extra structure'', and they ask whether such a characterization could be given for path graphs.
Here we argue that a characterization of this type is not possible.

First, let us examine\footnote{Our presentation here is slightly different than that of Cameron et al.\cite{4}} the argument by Cameron et al.
A \emph{pointed graph} is a tuple $(G, u, v)$, where $G$ is a graph and $u, v \in G$ are two distinguished, distinct vertices.
If $\mathcal{G} = (G, u, v)$ and $\mathcal{G}' = (G', u', v')$ are pointed graphs, then an \emph{isomorphism} $\phi$ between $\mathcal{G}$ and $\mathcal{G}'$ is an isomorphism $\phi$ from $G$ to $G'$ such that $\phi(u) = u'$ and $\phi(v) = v'$.
If $\mathcal{S}$ is a family of pointed graphs and $G$ is a graph with two distinguished vertices $u, v$, we say that $u$ and $v$ are \emph{$\mathcal{S}$-connected} in $G$ if there is an induced subgraph $H$ of $G$ containing $u$ and $v$ such that $(H, u, v)$ is isomorphic to some pointed graph in $\mathcal{S}$.
A set of vertices $T$ in $G$ is an \emph{$\mathcal{S}$-asteroidal set} if $T$ is an asteroidal set and every pair of vertices in $T$ is $\mathcal{S}$-connected.
The next class of pointed graphs was introduced in \cite{4}.

\begin{defn}
  The family $\mathcal{S}_{\mathsf{directed}}$ contains the following set of pointed graphs, each classified into one of four types.
  The two distinguished vertices will always be labelled $u, v$.
  (See Figure \ref{fig:specialconnect})
  \begin{description}
  \item[Type 1:] The $3$-path $(P_3, u, v)$, with endpoints $u, v$ distinguished.
  \item[Type 2:] The graph defined by vertices $\set{u, v, a, b, c, d}$ and edges $ua, uc, ab, bc, ac, bd, cd, bv, dv$.
  \item[Type 3:] The graph defined by vertices $u, v, a, b, c, d, x, y$, with triangles $\set{u,a,c}, \set{v, b, d}$, and 4-cliques $\set{a, b, c, d}$, $\set{x, a, b, c}$, and `$\set{y, a, b, d}$.
  \item[Type 4:] For any $t \geq 1$, the graph defined by vertices $u, v, z_0, \ldots, z_{2t+2}, z_1', \ldots, z_{2t}'$, where vertices $z_0, \ldots, z_{2t+2}$ form a clique and vertex $z_k'$ (for $0 \leq k \leq 2t+1$, with $z_0' = u, z'_{2t+1} = v$) is adjacent to $z_k$ and $z_{k+1}$.
  \end{description}
\end{defn}

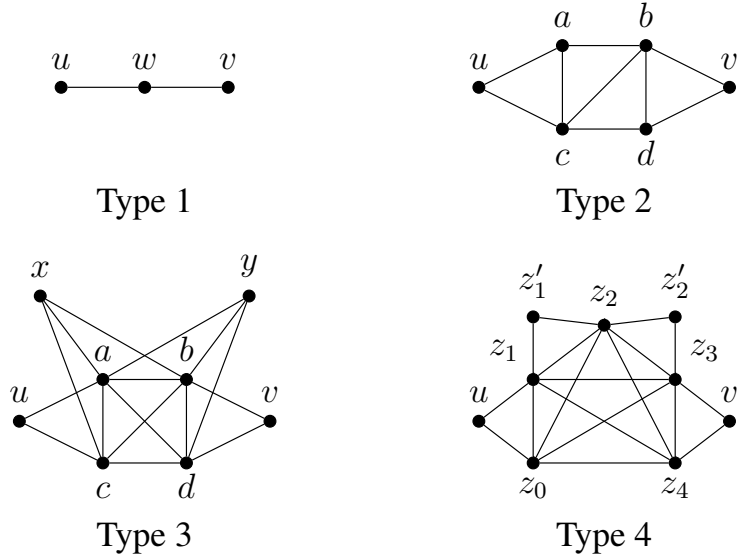
\begin{figure}[htbp]
  \centering
  \resizebox{0.65\linewidth}{!}{
    \begin{tikzpicture}
      \node[circle, draw, fill=black!100, inner sep=1pt, minimum width=4pt, label={$u$}] (u) at (0,0) {};
      \node[circle, draw, fill=black!100, inner sep=1pt, minimum width=4pt, label={$w$}] (w) at (1,0) {};
    \node[circle, draw, fill=black!100, inner sep=1pt, minimum width=4pt, label={$v$}] (v) at (2,0) {};
    \node[draw=none] at (1, -1.4) {Type 1};
    \foreach \from/\to in {u/w, w/v}
    {\draw (\from) -- (\to);}

    \pgfmathsetmacro{\xoffset}{5}
    \node[circle, draw, fill=black!100, inner sep=1pt, minimum width=4pt, label={$u$}] (u1) at (0+\xoffset,0) {};
    \node[circle, draw, fill=black!100, inner sep=1pt, minimum width=4pt, label={$a$}] (a1) at (1+\xoffset,0.5) {};
    \node[circle, draw, fill=black!100, inner sep=1pt, minimum width=4pt, label={[yshift=-0.65cm]$c$}] (c1) at (1+\xoffset,-0.5) {};
    \node[circle, draw, fill=black!100, inner sep=1pt, minimum width=4pt, label={$b$}] (b1) at (2+\xoffset,0.5) {};
    \node[circle, draw, fill=black!100, inner sep=1pt, minimum width=4pt, label={[yshift=-0.65cm]$d$}] (d1) at (2+\xoffset,-0.5) {};
    \node[circle, draw, fill=black!100, inner sep=1pt, minimum width=4pt, label={$v$}] (v1) at (3+\xoffset,0) {};
    \node[draw=none] at (1.5+\xoffset, -1.4) {Type 2};
    \foreach \from/\to in {u1/a1, u1/c1, a1/c1, a1/b1, c1/d1, c1/b1, b1/d1, b1/v1, d1/v1}
    {\draw (\from) -- (\to);}

    \pgfmathsetmacro{\xoffset}{-0.5}    
    \pgfmathsetmacro{\yoffset}{-4}
    \node[circle, draw, fill=black!100, inner sep=1pt, minimum width=4pt, label={$u$}] (u2) at (0+\xoffset,0+\yoffset) {};
    \node[circle, draw, fill=black!100, inner sep=1pt, minimum width=4pt, label={$a$}] (a2) at (1+\xoffset,0.5+\yoffset) {};
    \node[circle, draw, fill=black!100, inner sep=1pt, minimum width=4pt, label={[yshift=-0.65cm]$c$}] (c2) at (1+\xoffset,-0.5+\yoffset) {};
    \node[circle, draw, fill=black!100, inner sep=1pt, minimum width=4pt, label={$b$}] (b2) at (2+\xoffset,0.5+\yoffset) {};
    \node[circle, draw, fill=black!100, inner sep=1pt, minimum width=4pt, label={[yshift=-0.65cm]$d$}] (d2) at (2+\xoffset,-0.5+\yoffset) {};
    \node[circle, draw, fill=black!100, inner sep=1pt, minimum width=4pt, label={$v$}] (v2) at (3+\xoffset,0+\yoffset) {};
    \node[circle, draw, fill=black!100, inner sep=1pt, minimum width=4pt, label={$x$}] (x2) at (0.25+\xoffset,1.5+\yoffset) {};
    \node[circle, draw, fill=black!100, inner sep=1pt, minimum width=4pt, label={$y$}] (y2) at (2.75+\xoffset,1.5+\yoffset) {};
    \node[draw=none] at (1.5+\xoffset, -1.4+\yoffset) {Type 3};
    \foreach \from/\to in {u2/a2, u2/c2, a2/c2, a2/b2, c2/d2, c2/b2, b2/d2, b2/v2, d2/v2, a2/d2, x2/a2, x2/c2, x2/b2, y2/a2, y2/b2, y2/d2}
    {\draw (\from) -- (\to);}

    \pgfmathsetmacro{\xoffset}{4.5}    
    \pgfmathsetmacro{\yoffset}{-4.5}
    \node[circle, draw, fill=black!100, inner sep=1pt, minimum width=4pt, label={$u$}] (u3) at (0.5+\xoffset,0.5+\yoffset) {};
    \node[circle, draw, fill=black!100, inner sep=1pt, minimum width=4pt, label={$v$}] (v3) at (3.5+\xoffset,0.5+\yoffset) {};
    \node[circle, draw, fill=black!100, inner sep=1pt, minimum width=4pt, label={[yshift=-0.65cm]$z_0$}] (z0) at (1.15+\xoffset,0+\yoffset) {};
    \node[circle, draw, fill=black!100, inner sep=1pt, minimum width=4pt, label={[xshift=-0.35cm]$z_1$}] (z1) at (1.15+\xoffset,1+\yoffset) {};
    \node[circle, draw, fill=black!100, inner sep=1pt, minimum width=4pt, label={$z_2$}] (z2) at (2+\xoffset,1.65+\yoffset) {};
    \node[circle, draw, fill=black!100, inner sep=1pt, minimum width=4pt, label={[xshift=0.35cm]$z_3$}] (z3) at (2.85+\xoffset,1+\yoffset) {};
    \node[circle, draw, fill=black!100, inner sep=1pt, minimum width=4pt, label={[yshift=-0.65cm]$z_4$}] (z4) at (2.85+\xoffset,0+\yoffset) {};
    \node[circle, draw, fill=black!100, inner sep=1pt, minimum width=4pt, label={$z_1'$}] (z11) at (1.15+\xoffset,1.75+\yoffset) {};
    \node[circle, draw, fill=black!100, inner sep=1pt, minimum width=4pt, label={$z_2'$}] (z22) at (2.85+\xoffset,1.75+\yoffset) {};
    \node[draw=none] at (2+\xoffset, -0.9+\yoffset) {Type 4};
    \foreach \a in {0, ..., 4}{
      \foreach \b in {\a, ..., 4}{
        \draw (z\a) -- (z\b);
      }
    }
    \foreach \from/\to in {u3/z0, u3/z1, v3/z3, v3/z4, z11/z1, z11/z2, z22/z2, z22/z3}
    {\draw (\from) -- (\to);}
  \end{tikzpicture}
  }
  \caption{The family $\mathcal{S}_{\mathsf{directed}}$ of ``special connections'' \cite{4}.}
  \label{fig:specialconnect}
\end{figure}

The main result of Cameron et al.\cite{4} is the following:

\begin{thm}
  A chordal graph $G$ is a directed path graph if and only if $G$ does not contain an $\mathcal{S}_\mathsf{directed}$-asteroidal triple.
\end{thm}

At the end of their paper, Cameron et al. ask whether or not it is possible to give a characterization of path graphs using some ``special type'' of asteroidal triples.
Of course, if we use the framework outlined above, there is a simple (albeit unsatisfying) answer.
Define
\[ \mathcal{S}_{\mathsf{path}} = \set{(G, u, v) \st G \emph{ is a forbidden induced subgraph of path graphs and } u, v \in G}.\]
Clearly if a chordal graph $G$ contains an $\mathcal{S}_{\mathsf{path}}$-asteroidal triple then $G$ is not a path graph for the simple reason that $G$ would contain a forbidden induced subgraph of path graphs.
However, this argument is a ``cheat'' in that it is not using any asteroidal structure.

The next proposition gives strong evidence that it will be difficult to give \emph{any} nice characterization of path graphs using only asteroidal triples.
Intuitively, it says that there exists a family of chordal, non-path graphs $G$ which have many asteroidal triples, but for \emph{every} asteroidal triple there is a way to construct the clique tree for $G$ such that examining the asteroidal triple in the clique tree does not certify that $G$ is a path graph.
One such family is the family of subgraphs $F_{11}(4k)$ from Figure \ref{fig:FISC} for $k \geq 3$.
Figure \ref{fig:badgraph} below depicts the graph from this family corresponding to $k = 4$.

\begin{figure}[thb]
  \centering
  \includegraphics[scale=0.15, clip=true, trim=0 0 0 10mm]{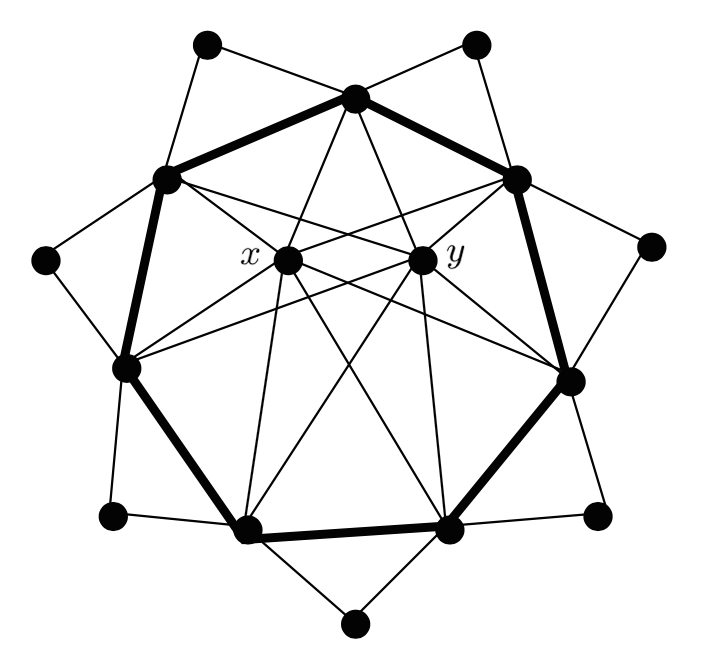}
  \caption{$F_{11}(16)$}
  \label{fig:badgraph}
\end{figure}

\begin{prop}\label{prop:at_obstruction}
  Consider $G = F_{11}(4k)$ for any $k \geq 3$ (cf.~Figure \ref{fig:FISC}) and let $A$ be any asteroidal triple in $G$.
  Then there exists a clique tree $T_A$ of $G$ such that for every $r \in A$ and every $v \in N(r)$, $T_A^v$ is a path in $T_A$.
\end{prop}
\begin{proof}
  Let $x, y$ denote the two central vertices in $G$, and let $R$ be the set of rays.
  Since $k \geq 3$ we know that $|R| > 3$, and clearly $G$ has $|R| + 2$ maximal cliques: the two ``central cliques'' $Q_x, Q_y$ containing $x$ and $y$, respectively, and a triangle $Q_r$ for each ray $r \in R$.
  It is easy to see that $A \subseteq R$, and note that at least two of the vertices $r, r' \in A$ satisfy $N(r) \cap N(r') = \emptyset$ since $|R| > 3$.
  We define the clique tree $T_A$ as follows.
  Connect the vertices $Q_x$ and $Q_y$.
  If $r, r' \in A$ are distinct and satisfy $N(r) \cap N(r') = \emptyset$, attach $Q_r$ and $Q_{r'}$ to $Q_x$, and otherwise connect $Q_r$ to $Q_x$ and $Q_{r'}$ to $Q_y$.
  Connect all other cliques $Q_{r^*}$ to $Q_x$ or $Q_y$ arbitrarily.
  Then for every $r \in A$ and $v \in N(r)$ it is easy to see that $T_A^v$ is a path.
\end{proof}

A chordal graph $G$ is not a path graph if and only if it does not have a clique-path tree --- or, in other words, if for any clique tree of $G$ there must be a vertex $v \in G$ such that $T^v$ is not a path.
Thus, if we were able to define some ``special type'' of asteroidal triple that obstructed the existence of a clique-path tree we would have to use this triple to show that in any clique tree $T$ of $G$ there is a vertex $v$ such that $T^v$ is not a path.
The previous proposition shows that if $G$ is a chordal, non-path graph, $A$ is any asteroidal triple in $G$, and $T$ is any clique-tree of $G$, then it is not sufficient to study the local structure of $A$ in $T$ in order to show that $G$ is not a path graph.
In the new characterization (cf.~Theorem \ref{thm:main}) the case of $F_{11}(4k)$ for $k \geq 3$ falls under Lemma \ref{lem:cycle-rays}, where the obstruction in the clique-tree is due to the parity argument that must examine all of the asteroidal vertices at once.



\section{Conclusion}
\label{sec:conclusion}

Theorem \ref{thm:main} gives a new characterization of path graphs by forbidding particular sun systems from chordal graphs.
Additionally, we give some strong evidence that asteroidal triples are not enough to capture the underlying structure of path graphs.

One of the main open problems with respect to path graphs is a linear time recognition algorithm. 
Both interval and chordal graphs have $O(m+n)$ time recognition algorithms \cite{13, 12}.
A linear time recognition algorithm for path graphs was proposed in \cite{15}, but is not considered correct (we refer the reader to Section 2.1.4 in \cite{14}). 
Chaplick et al.~showed that recognition of directed path graphs is no more difficult than the recognition of path graphs \cite{9}.
In \cite{7}, Gavril gave the first recognition algorithm for path graphs that runs in $O(n^4)$ time, and Sch\"{a}ffer \cite{16} then Chaplick \cite{14} gave $O(mn)$ time algorithms. 
We raise the question whether the structure of sun systems can lead to a faster than $O(mn)$ time recognition algorithm for path graphs. 

There are also several ``structural'' open problems suggested by our work.
Our characterization of path graphs should be contrasted with the characterization of directed path graphs by Chaplick et al. \cite{9} by forbidding odd suns. We raise the question of whether the maximum cardinality clique search algorithm developed in \cite{9} to recognize directed path graphs - and extract an odd sun if one exists - can prove useful for the recognition of sun-systems and thus of path graphs. 
Also, recall that the classical characterization of interval graphs by forbidden induced subgraphs --- given by Lekkerkerker and Boland \cite{6} --- proceeds by first using the characterization of interval graphs as chordal graphs which are asteroidal-triple free.
Since the ``bad'' sun systems have a particular structure on their asteroidal vertices, it is natural to ask whether we can use them to give a simplified proof of the characterization of path graphs by forbidden induced subgraphs \cite{3}.


\subsubsection*{Acknowledgments.} The authors thank Derek Corneil and Ekkehard K\"{o}hler for their helpful suggestions and valuable comments.
\bibliographystyle{splncs}

\begin{thebibliography}{}	
\bibitem{4}
	Cameron, K., Ho{\`{a}}ng, C.T., L{\'{e}}v{\^{e}}que, B.: 
	Characterizing directed path graphs by forbidden asteroids.
	Journal of Graph Theory 68(2), 103-112 (2011)

\bibitem{14}
	Chaplick, S.: 
	PQR-trees and undirected path graphs.
	M.Sc. Thesis, Dept. of Computer Science, University of Toronto, Canada (2008)
	
\bibitem{9}
	Chaplick, S., Gutierrez, M., L{\'{e}}v{\^{e}}que, B., Tondato, S.B.:
	From Path Graphs to Directed Path Graphs.
	36$^{th}$ Graph Theoretic Concepts in Computer Science, WG 2010, 256-265 (2010)
	
\bibitem{13}
	Corneil, D.G., Olariu, S., Stewart, L.:
	The LBFS structure and recognition of interval graphs.
	SIAM Journal on Discrete Mathematics, 23(4), 1905-1953 (2009)

\bibitem{15}
	Dahlhaus, E., Bailey, G.: 
	Recognition of path graphs in linear time. 
	5th Italian Conference on Theoretical Computer Science, 201-210 (1995)

\bibitem{11}
	Edmonds, J.:
	Paths, Trees, and Flowers.
	Canadian Journal of Mathematics 17(3), 449-467 (1965)	
	
\bibitem{10}
	Erd\"{o}s, P., Rado, R.:
	Intersection theorems for systems of sets.
	Journal of the London Mathematical Society, s1-35(1), 85-90 (1960)

\bibitem{5}
	Gavril, F.:
	The intersection graphs of subtrees in trees are exactly the chordal graphs.
	Journal of Combinatorial Theory, Series B 16(1), 47-56 (1974)

\bibitem{7}
	Gavril, F.:
	A recognition algorithm for the intersection graphs of paths in trees.
	Discrete Mathematics, 23(3), 211-227 (1978)

\bibitem{6}
	Lekkerkerker, C.G, Boland, J.C.:
	Representation of a finite graph by a set of intervals on the real line.
	Fundamenta Mathematicae 51, 45-64 (1962)	

\bibitem{3}
	L{\'{e}}v{\^{e}}que, B.,  Maffray, F., Preissmann, M.:
	Characterizing path graphs by forbidden induced subgraphs.
	Journal of Graph Theory 62(4), 369-384 (2009)		
	
\bibitem{1}
	Lin, I.J., McKee, T. A., West D.B.:
	The leafage of a chordal graph.
	Discussiones Mathematicae Graph Theory 18, 23-48 (1998)
	
\bibitem{8}
	Monma, C.L., Wei, V.K.:
	Intersection graphs of paths in a tree.
	Journal of Combinatorial Theory, Series B 41(2), 141-181 (1986)

\bibitem{2}
	Panda, B.S:
	The forbidden subgraph characterization of directed vertex graphs.
	Discrete Mathematics 196(1), 239-256(1999)
	
\bibitem{12}
	Rose, D.J., Tarjan, R.E., Lueker, G.S.:
	Algorithmic aspects of vertex elimination on graphs.
	SIAM Journal on Computing 5(2), 266-283 (1976)

\bibitem{16}
	Sch\"{a}ffer, A.A.: 
	A faster algorithm to recognize undirected path graphs. 
	Discrete Applied Math. (43), 261-295 (1993)
	
\end{thebibliography}

\end{document}